\documentclass[english, conference,  letter]{IEEEtran}

\usepackage{graphicx}

\usepackage[cmex10]{amsmath}
\usepackage{mathrsfs}
\usepackage[tight,footnotesize]{subfigure}
\usepackage{babel}
\usepackage{amsfonts}
\usepackage[latin1]{inputenc}
\usepackage{amssymb}
\usepackage{amsmath}

\newtheorem{theorem}{Theorem}

\newcounter{forex}[section]

\newtheorem{lemma}{Lemma}

\newtheorem{corollary}{Corollary}
\newtheorem{remark}{Remark}

\pdfoutput=1

\begin{document}
\title{A new bound on the capacity of the binary deletion channel with high deletion probabilities}


\author{
\authorblockN{Marco Dalai,}
\authorblockA{Department of Information Engineering\\University of Brescia, Italy\\Email: marco.dalai@ing.unibs.it}
}

\maketitle

\begin{abstract}
Let $C(d)$ be the capacity of the binary deletion channel with deletion probability $d$. It was proved by
Drinea and Mitzenmacher that, for all $d$,
$C(d)/(1-d)\geq 0.1185
$. Fertonani and Duman recently showed that $\limsup_{d\to 1}C(d)/(1-d)\leq 0.49$. In this paper, it is proved that  
$\lim_{d\to 1}C(d)/(1-d)$ exists and is equal to $\inf_{d}C(d)/(1-d)$. This result suggests the conjecture that the curve $C(d)$ my be convex in the interval $d\in [0,1]$. Furthermore,
using currently known bounds for $C(d)$, it leads to the upper bound $\lim_{d\to 1}C(d)/(1-d)\leq 0.4143$. \end{abstract}

\section{Introduction}
A binary deletion channel $W^d$ is defined as a binary channel that drops bits of the input sequence independently with probability $d$. Those bits that are not dropped simply pass through the channel unaltered.
While simple to describe, the deletion channel proves to be very difficult to analyze. Dobrushin (\cite{dobrushin_1967}) showed that for such a channel it is possible to define a capacity $C(d)$ and that a Shannon like theorem applies to this channel. However, no closed formula expression is known up to now for the capacity $C(d)$, and only upper and lower bounds are currently available (see \cite{diggavi_2006,drinea_2007,kanoria_2010,kalai_2010,fertonani_2010}).

For small values of $d$, it was recently independently proved in \cite{kanoria_2010} and \cite{kalai_2010} that $C(d)\approx 1- H(d)$, where $H(d)$ is the binary entropy function. For values of $d$ close to $1$, it is known (see \cite{drinea_2006,fertonani_2010}) that $C(d)$ satisfies 
\begin{equation}
0.1185\leq \liminf_{d\to 1}\frac{C(d)}{1-d}\leq \limsup_{d\to 1}\frac{C(d)}{1-d}\leq 0.49
\label{eq:infsupbounds}
\end{equation}

As far as the author knows, there is no result in the literature on the existence of  $\lim_{d\to 1}C(d)/(1-d)$.
In this paper, it is proved that the limit exists and, in particular, that
\begin{equation}
\label{eq:quasoconv}
\lim_{d\to 1}\frac{C(d)}{(1-d)}=\inf_{d}\frac{C(d)}{(1-d)}.
\end{equation}

The best currently known upper bound for $C(d)$, when used in the right hand side of \eqref{eq:quasoconv}, leads to the upper bound
\begin{equation}
\lim_{d\to 1}\frac{C(d)}{(1-d)}\leq 0.4143,
\label{eq:newbound}
\end{equation}
which improves the best previously known bound of equation \eqref{eq:infsupbounds}.
Furthermore, equation \eqref{eq:quasoconv} suggests the conjecture that $C(d)$ may be a convex function of $d$. Indeed, as discussed in Section \ref{sec:neard1} below, experimental evidence (see Figure \ref{fig:BaseBounds}) suggests the convexity of $C(d)$ for values of $d$ sufficiently smaller than $1$, while it is not easy to exclude that the function may be concave near $d=1$. Equation \eqref{eq:quasoconv} is only a necessary condition\footnote{It is not difficult to construct examples of ``pathological'' functions $f(d)$ that satisfy equation \eqref{eq:quasoconv}, when used in place of $C(d)$, but are not convex in any neighborhood of $d=1$.} for the convexity of $C(d)$ near $d=1$. It is, however, sufficient to conclude that $C(d)$ is not strictly concave in any neighborhood of $d=1$. Thus, either $C(d)$ exhibit a pathological behavior near $d=1$, or it is convex in a sufficiently small neighborhood of $d=1$.
A proof of the convexity of $C(d)$ would of course imply equation \eqref{eq:quasoconv} and thus equation \eqref{eq:newbound}.

The main idea used in this paper is the intuitive fact that, for a large enough number of input bits $n$, the deletion channel $W^d$ is fairly well approximated by a channel which drops exactly $[dn]$ bits selected uniformly at random. In particular, we show that a channel $W_{n,k}$ with $n$-bits input and $k$-bits output, selected uniformly within the $k$-bits subsequences of the input, has a capacity that is close to $C(1-k/n)$ for large enough $n$. Using this result, we build upon the work in \cite{fertonani_2010} to prove \eqref{eq:quasoconv}. 


\section{Definition and regularity of $C(d)$}
\label{sec:def}
For any $i$ and $j$, let $X_i^j=(X_i,X_{i+1},\ldots,X_j)$ and, similarly $Y_i^j=(Y_i,Y_{i+1}. \ldots,Y_j)$. 
Let $W_n^d$ be a channel with an $n$-bit string input whose output is obtained by dropping the bits of the input independently with probability $d$. Let then
\begin{equation}
C_n(d)=\frac 1 n \max_{p_{X_1^n}} I(X_1^n;W_n^d(X_1^n)).
\end{equation} 
It was proved by Dobrushin \cite{dobrushin_1967} that a transmission capacity $C(d)$ can be consistently defined for the deletion channel $W^d$ and that it holds 
\begin{equation}
\label{eq:limbase}
C(d)=\lim_{n \to \infty} C_n(d).
\end{equation}

\begin{figure}
\includegraphics[width=\linewidth]{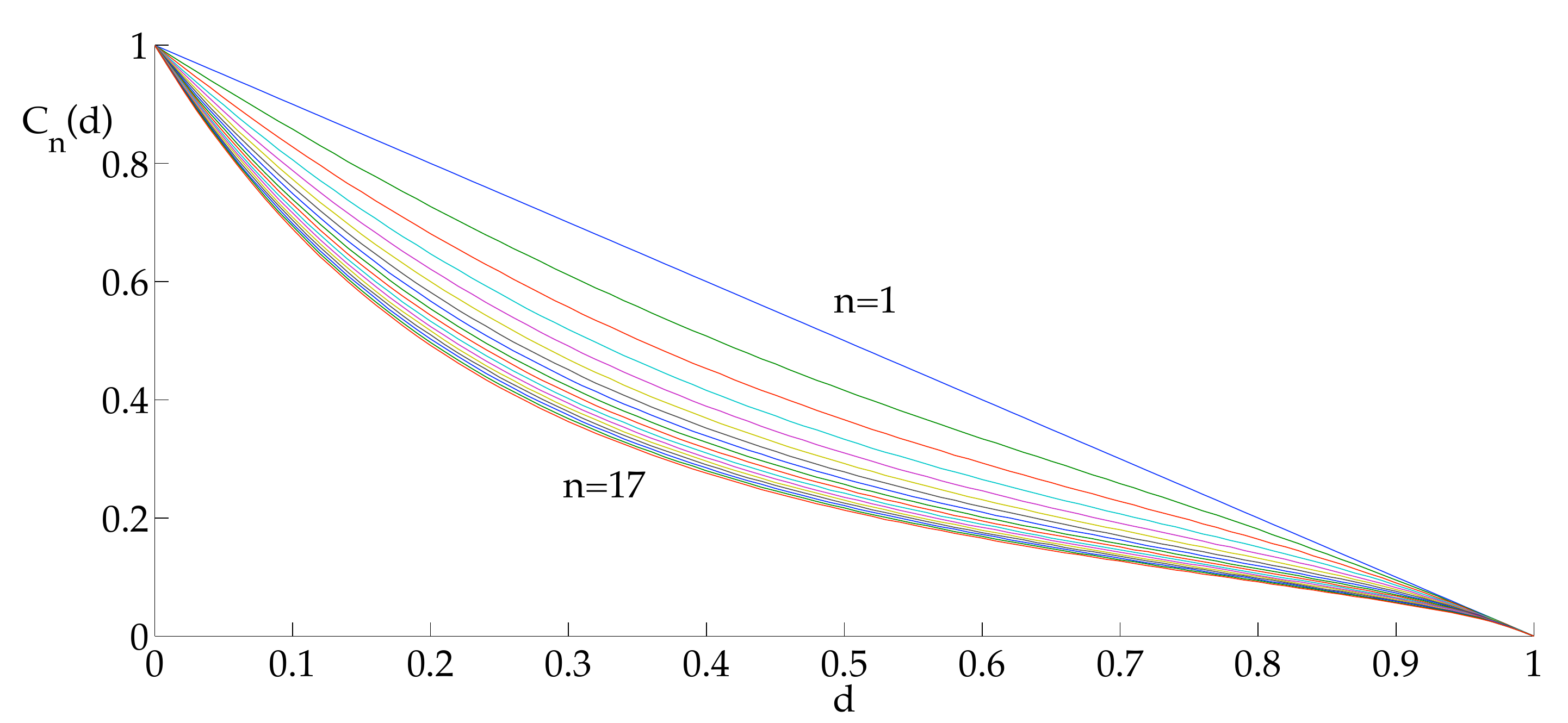}
\caption{Plot of the $C_n(d)$ functions for $n=1\ldots 17$ obtained by numerical evaluations in \cite{fertonani_2010}.}
\label{fig:BaseBounds}
\end{figure}

Figure \ref{fig:BaseBounds} shows the graph of the $C_n(d)$ functions for $n=1,\ldots,17$. The main objective of this section is to study the convergence of the $C_n(d)$ functions to deduce a regularity result for $C(d)$.

The following lemma gives a quantitative bound on the rate of convergence in \eqref{eq:limbase}.

\begin{lemma}
\label{l:boundCCn}(see also \cite{dobrushin_1967, kanoria_2010, fertonani_2010})
For every $d\in [0,1]$ and $n\geq 1$
\begin{equation}
C_n(d)-\frac{\log (n+1)}{n}\leq C(d)\leq C_n(d).
\end{equation}
\end{lemma}
\vspace{0.2cm}
\begin{proof}
As observed in \cite{kanoria_2010}, $nC_n(d)$ is a subadditive function of $n$. In fact, for an input $X_1^{n+m}$, let $\tilde{Y}_{(0)}=W_n^d(X_1^n)$  and $\tilde{Y}_{(1)}=W_m^d(X_{n+1}^{n+m})$. Note that $Y=W_{n+m}^d(X_1^{n+m})$ can be obtained as a concatenation of the strings $\tilde{Y}_{(0)}$ and $\tilde{Y}_{(1)}$. Thus, $X_1^{n+m} \to (\tilde{Y}_{(0)},\tilde{Y}_{(1)})\to Y$ is a Markov chain. Hence,
\begin{eqnarray*}
(n+m)C_{n+m}(d) & = & \max_{p_{X_1^{n+m}}}I(X_1^{n+m}; Y)\\
 & \leq & \max_{p_{X_1^{n+m}}} I(X_1^{n+m}; (\tilde{Y}_{(0)},\tilde{Y}_{(1)}))\\ & \leq  & nC_n(d)+mC_m(d).
\end{eqnarray*}
This implies by Fekete's lemma (see \cite[Prob. 98]{Polya-Szego}) that the limit $C(d)=\lim_{n\to \infty}C_n(d)$ exists and it satisfies $C(d)=\inf_{n\geq 1}C_n(d)$. This proves the right hand side inequality.

Take now an integer $h>1$ and consider, for an input $X_1^{hn}$, the output $Y=W_{hn}^d(X_1^{hn})$ as the concatenation of the $h$ outputs $\tilde{Y}_{(i)}=W_n^d(X_{ni+1}^{ni+n})$, $i=0,\ldots,h-1$. Let for convenience $\tilde{Y}_{(0)}^{(h-1)}=(\tilde{Y}_{(0)},\tilde{Y}_{(1)},\ldots,\tilde{Y}_{(h-1)})$.
 It is clear that $X_1^{hn} \to \tilde{Y}_{(0)}^{(h-1)}\to Y$ is a Markov Chain. Let $L_i$ be the length of $\tilde{Y}_{(i)}$. We thus have
\begin{eqnarray*}
hn C_{hn}(d) & =& \max_{p_{X_1^{hn}}}I(X_1^{hn}; Y)\\
& = & \max_{p_{X_1^{hn}}} [I(X_1^{hn}; \tilde{Y}_{(0)}^{(h-1)})
- I(X_1^{hn}; \tilde{Y}_{(0)}^{(h-1)}|Y)]\\
& \geq &  \max_{p_{X_1^{hn}}} [I(X_1^{hn}; \tilde{Y}_{(0)}^{(h-1)}) - H(\tilde{Y}_{(0)}^{(h-1)}|Y)]\\
& = &  \max_{p_{X_1^{hn}}} [I(X_1^{hn}; \tilde{Y}_{(0)}^{(h-1)})  - H(L_0^{h-1}|Y)]\\
& \geq &  \max_{p_{X_1^{hn}}} I(X_1^{hn}; \tilde{Y}_{(0)}^{(h-1)}) - (h-1)\log( n+1)\\
& = &  hnC_n(d)-(h-1)\log (n+1).
\end{eqnarray*}
Hence
\begin{eqnarray*}
C(d) & = & \lim_{h\to \infty}C_{hn}(d) \\
& \geq & \lim_{h\to \infty}\left[ C_n(d)-\frac{h-1}{h}\frac{\log(n+1)}{n}\right]\\
& = & C_n(d)-\frac{\log (n+1)}{n}.
\end{eqnarray*}
\end{proof}
See \cite[eq. (39)]{fertonani_2010} for tighter, though more complicated, bound.

As a consequence of Lemma \ref{l:boundCCn} we have the following regularity result for $C(d)$.
\begin{lemma}
\label{l:unifcont}
The function $C(d)$ is uniformly continuous in $[0,1]$. Thus, for every $\beta>0$ there is a $\alpha=\alpha(\beta)$ such that $|d_1-d_2|<\alpha\Rightarrow |C(d_1)-C(d_2)|<\beta$.
\end{lemma}
\vspace{0.2cm}
\begin{proof}
As shown in Lemma \ref{l:boundCCn}, the functions $C_n(d)$ tend to $C(d)$ uniformly in $d$. Hence, if proved that the  $C_n(d)$ are continuous in $d$, so is their limit $C(d)$. Since the domain of $C(d)$ is compact, by the Heine-Cantor theorem $C(d)$ is also uniformly continuous. That the $C_n(d)$ functions are continuous is really intuitive; the shortest formal proof that we were able to provide goes as follows. The entries of the transition matrix of the channel $W_n^d$ are polynomials in $d$ and thus the mutual information $I(X_1^n;W_n^d(X_1^n))$ is a continuous function of $d$ and of the input distribution $p_{X_1^n}$. Hence, by moving $d$ continuously from $0$ to $1$ one expects the capacity to change continuously from $1$ to $0$. A formal proof, however, seems to require using the compactness of the sets of distributions $p_{X_1^n}$.  Assume that $C_n(d)$ is not continuous in $d=\bar{d}$ and let $\bar{p}$ be the input distribution that attains the value $C_n(\bar{d})$. Then there exists an $\varepsilon>0$ such that $|C_n(\bar{d})-C_n(d_k)|>\varepsilon$ for a sequence $d_k$ converging to $\bar{d}$. Consider the distributions ${p_k}$ that attain $C_n(d_k)$. Since the set of the $p_{X_1^n}$ is bounded and closed, there exists a subsequence of the $p_k$ that converges to a distribution $p'$. By continuity of the mutual information the $C_n(d_k)$ values tend to the mutual information $I'$ attained by $p'$ in $d=\bar{d}$. But, by definition of $C_n(\bar{d})$, we clearly have that $I'\leq C_n(\bar{d})$ and thus $C_n(d_k)\leq C_n(\bar{d}) -\varepsilon $ for $k$ large enough. But then the mutual information attained by $\bar{p}$ in $d_k$ tends to $C_n(\bar{d})\geq C_n(d_k)+\varepsilon$ for large enough $k$, which is absurd by definition of $C_n(d_k)$.\end{proof}

\section{Exact deletion channel}
\label{sec:extactch}
Let now $W_{n,k}$, $k\leq n$, be a channel with $n$-bits input whose output is uniformly chosen within the $ n \choose k$ $k$-bits subsequences of the input. This channel was efficiently used as an auxiliary channel in \cite{kalai_2010,fertonani_2010}.
Let then
\begin{equation}
C_{n,k}=\frac 1 n \max_{p_{X_1^n}} I(X_1^n;W_{n,k}(X_1^n)).
\end{equation} 

The following obvious result will be used later.

\begin{lemma}
\label{l:k1k2}
For every random $X_1^n$, if $k_1\geq k_2$ then
\begin{equation}
I(X_1^n;W_{n,k_1}(X_1^n))\geq I(X_1^n;W_{n,k_2}(X_1^n)).
\end{equation}
\end{lemma}
\vspace{0.2cm}
\begin{proof}
Simply note that the $W_{n,k_2}$ channel can be obtained as a cascade of $W_{n,k_1}$ and $W_{k_1,k_2}$. Thus, $X_1^n\to W_{n,k_1}(X_1^n)\to W_{n,k_2}(X_1^n)$ is a Markov chain and the lemma follows from the data processing inequality. 
\end{proof}

The following lemma bounds the capacity of the $W_n^d$ channel in terms of the capacity of certain exact deletion channels.

\begin{lemma}
\label{l:cndbound}
For every $\varepsilon>0$, $d\in[\varepsilon,1-\varepsilon]$, and $n\geq 1$
\begin{equation}
C_{n,\lceil (1-d-\varepsilon)n \rceil}- \, 2 e^{-2 \varepsilon^2 n} \leq  C_n(d) \leq C_{n,\lfloor(1-d+\varepsilon)n \rfloor}+ 2e^{-2 \varepsilon^2 n}.
\end{equation}
\end{lemma}
\begin{proof}
We first prove the right hand side inequality. For an input $X_1^n$, let $Y=W_n^d(X_1^n)$ and let $L=|Y|$ be the length of $Y$. First note that $X_1^n \to Y \to L$ is a Markov chain. So, by applying the chain rule to $I(X_1^n;Y,L)$, considered that $I(X_1^n;L)=0$ since $L$ is independent from $X_1^n$, it is easily seen that $I(X_1^n;Y)=I(X_1^n;Y|L)$. Define $T=\{j: \left|\frac{j}{n}-(1-d)\right|\leq \varepsilon\}$, that is $j\in T$ if and only if $\lceil(1-d-\varepsilon)n\rceil\leq j\leq \lfloor(1-d+\varepsilon)n\rfloor$. Let  now $X_1^n$ be distributed according to the optimal distribution for the $W_n^d$ channel. Then we have
\begin{eqnarray*}
n C_n(d) & = &  I(X_1^n;Y | L) \\
 & = & \sum_{j=0}^n p_L(j) I(X_1^n;Y|L=j)\\
 & = & \sum_{j\in T} p_L(j)  I(X_1^n;Y|L=j)
 \\ & & \hspace{2cm}  
 + \sum_{j\in \bar{T}} p_L(j)  I(X_1^n;Y|L=j)\\
 & \stackrel{(a)}{\leq} & \sum_{j\in T} p_L(j)  I(X_1^n;Y|L=\lfloor(1-d+\varepsilon)n\rfloor)
 \\ & & \hspace{2cm}  
 + \sum_{j\in \bar{T}} p_L(j)  n\\
 & \leq & nC_{n,\lfloor(1-d+\varepsilon)n\rfloor} \sum_{j\in T} p_L(j)    + n\sum_{j\in \bar{T}} p_L(j)\\
 & \stackrel{(b)}{\leq} &  nC_{n,\lfloor(1-d+\varepsilon)n\rfloor}  +  2n e^{-2 \varepsilon^2n},
\end{eqnarray*}
where $(a)$ follows from Lemma \ref{l:k1k2} and the definition of $T$ and $(b)$ follows from the Chernoff bound. Dividing by $n$ we get the desired inequality.

As for the left hand side inequality, let now $X_1^n$ be distributed according to the optimal distribution for the $W_{n,\lceil(1-d-\varepsilon)n\rceil}$ channel. Then we have
\begin{eqnarray*}
n C_n(d) & \geq &  I(X_1^n;Y | L) \\
 & = & \sum_{j=0}^n p_L(j) I(X_1^n;Y|L=j)\\
 & = & \sum_{j\in T} p_L(j)  I(X_1^n;Y|L=j)
 \\ & & \hspace{2cm}  
 + \sum_{j\in \bar{T}} p_L(j)  I(X_1^n;Y|L=j)\\
 & \stackrel{(a)}{\geq} & \sum_{j\in T} p_L(j)  I(X_1^n;Y|L=\lceil(1-d-\varepsilon)n\rceil)\\
 & = & nC_{n,\lceil(1-d-\varepsilon)n\rceil} \sum_{j\in T} p_L(j)\\
 & \stackrel{(b)}{\geq} & nC_{n,\lceil(1-d+\varepsilon)n\rceil}(1-2 e^{-2 \varepsilon^2 n}) \\
 & \stackrel{(c)}{\geq} & nC_{n,\lceil(1-d+\varepsilon)n\rceil}-2n e^{-2\varepsilon^2n},
\end{eqnarray*}
where $(a)$ follows again from Lemma \ref{l:k1k2}, $(b)$ follows from the Chernoff bound, and $(c)$ follows from the obvious fact that $C_{n,\lceil(1-d+\varepsilon)n\rceil}\leq 1$. Dividing by $n$ the desired result is obtained.
\end{proof}

The following lemma bounds the capacity of the exact deletion channel $W_{n,k}$ in terms of $C(d)$ for appropriate values of $d$.

\begin{lemma}
For every $\varepsilon>0$ and integers $n$ and $k$
\begin{multline}
\label{eq:cnkbound}
C\left(1-k/n+\varepsilon\right) - 2e^{-2 \varepsilon^2 n} \leq C_{n,k} \leq C\left(1-k/n-\varepsilon\right) 
 \\
+ 2e^{-2 \varepsilon^2 n} + \frac{\log( n+1) }{n}.
\end{multline}
\end{lemma}
\vspace{0.2cm}
\begin{proof}
Take $d=1-k/n-\varepsilon$ in Lemma \ref{l:cndbound} to obtain $C_{n,k}\leq C_n(1-k/n-\varepsilon)+2e^{-2 \varepsilon^2 n}\leq C(1-k/n-\varepsilon)+2e^{-2 \varepsilon^2 n} +\log(n+1) /n$, by virtue of Lemma \ref{l:boundCCn}. Then take $d=1-k/n+\varepsilon$ in Lemma \ref{l:cndbound} to obtain $C_{n,k}\geq C_n(1-k/n+\varepsilon)-2e^{-2 \varepsilon^2 n}\geq C(1-k/n+\varepsilon)-2e^{-2 \varepsilon^2 n}$.
\end{proof}

\begin{lemma}
\label{l:absboundcnk}
For every $\beta>0$, there is an $\bar{n}=\bar{n}(\beta)$ such that
\begin{equation}
|C_{n,k}-C(1-k/n)|<\beta \qquad \forall n\geq \bar{n}, \,k=1,\ldots,n.
\end{equation}
\end{lemma}
\vspace{0.2cm}
\begin{proof}
First note that, for $\varepsilon>0$, $C(1-k/n+\varepsilon)\leq C(1-k/n)\leq C(1-k/n-\varepsilon)$. Hence, $C(1-k/n)$ satisfies the two inequalities satisfied by $C_{n,k}$ in equation \eqref{eq:cnkbound}. So, $|C_{n,k}-C(1-k/n)|$ is bounded by the difference between the right hand side and the left hand side of equation \eqref{eq:cnkbound}, that is
\begin{multline}
|C_{n,k}-C(1-k/n)|  \leq 
C\left(1-k/n-\varepsilon\right)-C\left(1-k/n+\varepsilon\right) 
 \\
 + 4e^{-2 \varepsilon^2 n} + \frac{\log (n+1) }{n}.
\end{multline}
With the notation of Lemma \ref{l:unifcont}, take $\varepsilon<\alpha(\beta/2)/2$ so that $C\left(1-k/n-\varepsilon\right)-C\left(1-k/n+\varepsilon\right)<\beta/2$. Once $\varepsilon$ is fixed, choose $\bar{n}$ such that $4e^{-2\varepsilon^2\bar{n}} + \frac{\log( \bar{n}+1) }{\bar{n}}<\beta/2$ to complete the proof. Note that $\bar{n}$ is a function of $\beta$ only and that the result holds for every $k\leq n$.
\end{proof}

We can now state the first result of this paper.
\begin{theorem} 
\label{l:limitnk}
Let $k_n$ be an integer valued sequence such that $k_n/n$ tends to $1-d$ as $n$ goes to infinity. Then
\begin{equation}
\lim_{n\to\infty} C_{n,k_n}=C(d).
\end{equation}
\end{theorem}
\vspace{0.2cm}
\begin{proof}
It follows easily from Lemma \ref{l:absboundcnk} by continuity of $C(d)$.
\end{proof}

\section{Behavior near $d=1$}
\label{sec:neard1}
In this Section, we finally focus on the behavior of the function $C(d)$ for values of $d$ close to 1.
It is interesting to observe in Figure \ref{fig:BaseBounds} that, from experimental evidence, the $C_n(d)$ functions seem to be convex in a progressively expanding region of $d$ values. On the one hand, it is tempting to conjecture that the limit $C(d)$ is convex in the whole interval $d\in [0,1]$.
On the other hand, near $d=1$, all the $C_n(d)$ curves appear to change concavity and go to zero asymptotically as $(1-d)$. Indeed, we have the following result.
\begin{lemma}
\label{lemma:cnd1}
For every $n$, 
\begin{equation}
\lim_{d\to 1} \frac{C_n(d)}{(1-d)}=1
\end{equation}
\end{lemma}
\vspace{0.2cm}
\begin{proof}
It is easily shown that for every $n$ and $d$
\begin{equation}
(1-d^n)/n\leq C_n(d) < (1-d).
\label{eq:boundcn}
\end{equation}
The right hand side inequality  follows from the fact that the capacity of $W_n^d$ is obviously smaller than the capacity of a binary erasure channel with erasure probability $d$. To prove the left hand side inequality consider using as input to the channel $W_n^d$ only the sequence composed of $n$ zeros and that composed of $n$ ones. Then the $n$ uses of $W_n^d$ correspond to one use of an erasure channel with erasure probability $d^n$. This proves equation \eqref{eq:boundcn}. Dividing by $(1-d)$ and taking the limit $d\to 1$ gives the required result.
\end{proof}

Lemma \ref{lemma:cnd1} ensures that, for fixed $n$, $C_n(d)$ is not convex in a neighborhood of $d=1$. Note further that 
\begin{equation}
\lim_{d\to 1} \frac{C_n(d)}{(1-d)}=\sup_{d\in (0,1)} \frac{C_n(d)}{(1-d)}=1
\end{equation}
Hence, it is natural to believe that $C_n(d)$ is actually concave in a neighborhood of $d=1$, even if Lemma \ref{lemma:cnd1} is not sufficient to prove this.
However, in the limit $n\to\infty$, it is known (see \cite{drinea_2006,fertonani_2010}) that $C(d)$ satisfies 
\begin{equation}
0.1185\leq \liminf_{d\to 1}\frac{C(d)}{1-d}\leq \limsup_{d\to 1}\frac{C(d)}{1-d}\leq 0.49
\end{equation}
Hence, Lemma \ref{lemma:cnd1} does not hold with $C(d)$ in place of $C_n(d)$ and it is still legitimate to conjecture that $C(d)$ may be convex in $[0,1]$.
The next step is thus to ask if $C_n(d)/(1-d)$ has a limit as $d\to 1$ and, if so, if this limit is reached from above as would be implied by convexity of $C(d)$. The remaining part of this section tries to answer this question.

In order to understand the behavior of $C(d)$ near $d=1$, the following result from \cite{fertonani_2010} is fundamental.
\begin{lemma}[Fertonani and Duman, {\cite[eq. (32)]{fertonani_2010}}]
\label{l:fertonani}
For every $n,k$
\begin{equation}
\limsup_{d\to 1}\frac{C(d)}{1-d}\leq \frac{n C_{n,k}+1}{k+1}.
\label{eq:fertonani}
\end{equation}
\end{lemma}
\vspace{0.3cm}
\begin{remark}
In \cite{fertonani_2010} the authors state that, for every $n$ and $k$, $\lim_{d\to 1} \frac{C(d)}{1-d}\leq \frac{n C_{n,k}+1}{k+1}$. However, we are not aware of a previous formal proof that $\lim_{d\to 1} \frac{C(d)}{1-d}$ exists. This fact is proved in the following theorem.
\end{remark}

\begin{theorem}
\label{th:limeqinf}
It holds that
\begin{equation}
\label{eq:limdiseq}
\lim_{d\to 1}\frac{C(d)}{(1-d)} = \inf_{d\in(0,1)}\frac{C(d)}{1-d}.
\end{equation}
\end{theorem}
\vspace{0.2cm}
\begin{proof}
For every $d'\in (0,1)$, let $k_n$ be a sequence such that $k_n/n$ tends to $1-d'$. Then, from Theorem \ref{l:limitnk}, the right hand side of \eqref{eq:fertonani}, with $k_n$ in place of $k$, tends to $C(d')/(1-d')$. Since $d'$ is arbitrary, Lemma \ref{l:fertonani} implies that $\limsup_{d\to 1}C(d)/(1-d) \leq \inf_{d'\in(0,1)}\frac{C(d')}{1-d'}$. However, it is obvious that $\liminf_{d\to 1}C(d)/(1-d) \geq \inf_{d'\in(0,1)}\frac{C(d')}{1-d'}$. Thus $\lim_{d\to 1}C(d)/(1-d)$ exists and is equal to $ \inf_{d'\in(0,1)}\frac{C(d')}{1-d'}$
\end{proof}

A direct consequence of Theorem \ref{th:limeqinf} is the following improved bound on $C(d)$.
\begin{corollary}
\label{cor:newbound}
\begin{equation}
\lim_{d\to 1}\frac{C(d)}{(1-d)}\leq 0.4143.
\end{equation}
\end{corollary}
\vspace{0.2cm}
\begin{proof}
As far as the author knows, the best known numerical bound obtained for $\inf_d C(d)/(1-d)$ is $0.4143$ obtained using the bound $C(0.65)\leq C_{17}(0.65)=0.145$, numerically evaluated in  \cite{fertonani_2010}. 
\end{proof}

The usefulness of Theorem \ref{th:limeqinf} is that it allows to deduce provable bounds for $\lim_{d\to 1}\frac{C(d)}{(1-d)}$ from bounds on $C(d)$ even with $d$ much smaller than 1. It is interesting to note, in fact, that different techniques seem to be effective in bounding $C(d)$ in different regions of the interval $[0,1]$. For example, different genie aided channels are used in \cite{fertonani_2010} for smaller values of $d$ than for large values of $d$ and, while equation \eqref{eq:fertonani} is derived in \cite{fertonani_2010} using a bound effective for large $d$, the bound for $C(0.65)$ used in Corollary \ref{cor:newbound} is derived from the numerical value of $C_{17}(d)$ which is not as effective for $d$ larger than $0.8$ (see Table IV in \cite{fertonani_2010}, where bound $C_4$ therein is what we called $C_{17}(d)$, while bound $C_2^*$ is used to deduce \eqref{eq:fertonani}). Thus, in order to obtain improved upper bounds for $\lim_{d\to 1}\frac{C(d)}{(1-d)}$ one effective approach would be to numerically evaluate $C_n(d)$ near $d=0.65$ for $n\geq 18$. This requires, however, high computational and spatial complexity and it is out of the scope of the present paper.

\section{acknowlegdments}
The author would like to thank Dario Fertonani for useful discussions and for providing numerical data obtained during the preparation of \cite{fertonani_2010}.


\begin{thebibliography}{1}

\bibitem{dobrushin_1967}
R.~L. Dobrushin,
\newblock ``Shannon's theorems for channels with synchronization errors,''
\newblock {\em Problems of Information Transmission}, vol. 3, no. 4, pp.
  11--26, 1967.

\bibitem{diggavi_2006}
S.~Diggavi and M.~Grossglauser,
\newblock ``On information transmission over a finite buffer channel,''
\newblock {\em IEEE Trans. Inform. Theory}, vol. 52, no. 3, pp. 1226--1237,
  2006.

\bibitem{drinea_2007}
M.~Drinea and M.~Mitzenmacher,
\newblock ``Improved lower bounds for the capacity of i.i.d. deletion and
  duplication channels,''
\newblock {\em IEEE Trans. on Inform. Theory}, vol. 53, no. 8, pp. 2693--2714,
  2007.

\bibitem{kanoria_2010}
Y.~Kanoria and A.~Montanari,
\newblock ``On the deletion channel with small deletion probability,''
\newblock submitted.

\bibitem{kalai_2010}
A.~Kalai, M.~Mitzenmacher, and M.~Suda,
\newblock ``Tight asymptotic bounds for the deletion channel with small
  deletion probabilities,''
\newblock in {\em Proc. IEEE Intern. Symp. on Inform. Theory}, 2010.

\bibitem{fertonani_2010}
D.~Fertonani and T.~M. Duman,
\newblock ``Novel bounds on the capacity of the binary deletion channel,''
\newblock {\em IEEE Trans. Inform. Theory}, vol. 56, no. 6, pp. 2753--2765,
  June 2010.

\bibitem{drinea_2006}
M.~Drinea and M.~Mitzenmacher,
\newblock ``A simple lower bound for the capacity of the deletion channel,''
\newblock {\em IEEE Trans. on Inform. Theory}, vol. 52, no. 10, pp. 4657--4660,
  2006.

\bibitem{Polya-Szego}
G.~P\'{o}lya and G.~Szeg\"{o},
\newblock {\em Problems and Theorems in Analysis}, vol.~1,
\newblock Springer-Verlag, 1976.

\end{thebibliography}
\end{document}